\documentclass[11pt]{article}

\usepackage[english]{babel}

\usepackage{amsmath}
\usepackage{amssymb}
\usepackage{amsthm}
\usepackage{fancyhdr}
\usepackage{verbatim}
\usepackage{graphicx}

\newcommand{\MP}{\mathbb{P}}

\newcommand{\beq}{\begin{equation}}
\newcommand{\eeq}{\end{equation}}
\newcommand{\bi}{\begin{itemize}}
\newcommand{\bd}{\begin{description}}
\newcommand{\ei}{\end{itemize}}
\newcommand{\ed}{\end{description}}

\newcommand{\bc}{\begin{center}}
\newcommand{\ec}{\end{center}}

\newtheorem{Thm}{\bf Theorem}[section]
\newtheorem{Cor}[Thm]{\bf Corollary}
\newtheorem{Prop}[Thm]{\bf Proposition}
\newtheorem{Lem}[Thm]{\bf Lemma}

\newtheorem{Rem}[Thm]{\bf Remark}

\numberwithin{equation}{section}

\title{{\bf An application to credit risk of a hybrid Monte Carlo-Optimal quantization method}}

\author{ \vspace{1 cm}\\
{\sc \Large Giorgia Callegaro} \thanks{Scuola Normale Superiore di Pisa, Piazza dei Cavalieri 7, I-56100 Pise, Italy. Universit\'{e} d'\'{E}vry Val d'Essonne, D\'{e}partement de Math\'{e}matiques, Bd. F. Mitterrand,  F-91025 \'{E}vry Cedex. Ecole Nationale Sup\'{e}rieure d'Informatique pour l'Industrie et l'Entreprise,
1 square de la R\'{e}sistance,
91025 Evry Cedex. Phone +33 (0)1 69 47 02 05, e-mail: {\tt g.callegaro@sns.it. } The author's research is supported by an AMaMeF exchange grant and the ``Chaire Risque de Cr\'{e}dit'' of the French Banking Federation. }
\vspace{0.2 cm}\\
{\it Scuola Normale Superiore di Pisa, }\vspace{0.1 cm}\\
{\it Universit\'{e} d'\'{E}vry Val d'Essonne \&} \vspace{0.15 cm} \\
{\it ENSIIE } \vspace{1,5 cm}\\
{\sc \Large Abass Sagna}   \thanks{Laboratoire de Probabilit\'es et Mod\`eles Al\'eatoires, UMR 7599, Universit\'e Pierre et Marie Curie (Paris VI), 175 Rue du Chevaleret, 75013 Paris. Phone +33 (0)1 44 27 72 22, e-mail: {\tt abass.sagna@gmail.com}.}
\vspace{0.2 cm}\\
{\it Universit\'e Pierre et Marie Curie, Paris VI}
}

\date{\vskip 40 pt  \today }

\begin{document}

\maketitle
\newpage

\begin{abstract}
\vspace{0.5 cm}
In this paper we use a hybrid Monte Carlo-Optimal quantization method to approximate the conditional survival probabilities of a firm, given a structural model for its credit default, under partial information.

We consider the case when the firm's value is a non-observable stochastic process ${(V_t)}_{t \ge 0}$ and investors in the market have access to a process ${(S_t)}_{t \ge 0}$, whose value at each time $t$ is related to $(V_s, s \le t)$.
We are interested in the computation of the conditional survival probabilities of the firm given the ``investor information''.


As an application, we analyze the shape of the credit spread curve for zero coupon bonds in two examples.
\end{abstract}
\vspace{0.5 cm}

\newpage
\section*{Introduction}
In this paper we compute the conditional survival probabilities of a firm, in a market that is not transparent to bond investors, by using both Monte Carlo and optimal quantization methods. This allows us to analyze the credit spread curve under partial information in some examples, in order to investigate the degree of transparency and riskiness of a firm, as viewed by bond-market participants.

To introduce the problem, recall that most of the bonds traded in the market are corporate bonds and treasury bonds, that are consequently subject to many kinds of risks, such as market risk (due for example to changes in the interest rate), counterparty risk and liquidity risk. One of the main challenges in credit risk modeling is, then, to quantify the risk associated to these financial instruments.

The methodology for modeling a credit event can be split into two main approaches: the structural approach, introduced by Merton in 1974 and the reduced form approach (or ``intensity based''), originally developed by Jarrow and Turnbull in 1992.

The structural approach consists in modeling the credit event as the first hitting time of a barrier by the firm value process.

In reduced form models the default intensity is directly modeled and it is given by a function of latent state variables or predictors of default.

The first approach, in which we are interested, is intuitive by the economic point of view, but it presents some drawbacks: the firm value process can not be easily observed in practice, since it is not a tradeable security, and a continuous firm's value process implies a predictable credit event, leading to unnatural and undesirable features, such as null spreads for surviving firms for short maturities.

Despite the apparent difference between the two models (see, e.g., Jarrow and Protter, 2004), some recent results, starting from the seminal paper Duffie and Lando (2001), have unified the two approaches by means of information reduction. We consider, then, a structural model under partial information, in which investors can not observe the firm value process, but they have access to another process whose value is related to the firm value process.
We show in two examples that yield spreads for surviving firms are strictly positive at zero maturity, since investors are uncertain about the nearness of the current firm value to the trigger level at which the firm would declare
default.
The shape of the term structure of credit spreads may be useful, then, in practice to estimate the degree of transparency and of riskiness of a firm, from the investors' point of view.

We show that the computation of the conditional survival probabilities under partial information leads to a nonlinear filtering problem involving the conditional survival probabilities under full information. These former quantities are ap\-pro\-xi\-ma\-ted by a Monte Carlo procedure, while the filter distribution is estimated by optimal quantization.

The paper is organized as follows. In the first section, we present the market model and we decompose our problem into two problems (\textbf{P1}) and (\textbf{P2}), that are, respectively, the computation of conditional survival probability in a full information setting and the approximation of the filter distribution. Section \ref{ApproxFilter} and Section \ref{PsurvieMC} are devoted to the solution of the previous two problems. We provide error estimates in Section \ref{Error} and, finally, in Section \ref{NumRes} we present two numerical examples concerning the application to credit risk.

\section{Market model and problem definition}\label{S.2}
Let us consider a probability space $(\Omega, \mathcal F,\mathbb P)$, representing all the randomness of our economic context.\\
For the moment we concentrate our attention on the ``real world'' probability measure $\mathbb P$ and on a single firm model, in which the company is subject to default risk and we use a structural approach to characterize the default time.\\
The process representing the value of the firm, given for example by its value of financial statement, is denoted by ${(V_t)}_{t \ge 0}$ and we suppose it can be modeled as the solution of the following stochastic differential equation
\begin{equation}
\label{dVt}
\left \{
\begin{array}{rcl}
 d V_t &=& b(t,V_t)  d t + \sigma(t,V_t)   d W_t,\\ \quad V_0 & = & v_0,
\end{array}
\right.
\end{equation}
where the functions $b:[0,+ \infty) \times \mathbb R \rightarrow \mathbb R$ and $\sigma:[0,+ \infty) \times \mathbb R \rightarrow \mathbb R$ are 
Lipschitz in $x$ uniformly in $t$ and $W$ is a standard one-dimensional Brownian motion. We suppose that $\sigma(t,x)>0$ for every $(t,x) \in [0, + \infty) \times \mathbb{R} $.  \\
In our setting the process $V$ is non observable (it is also known as \emph{state} or \emph{signal}), but investors have access to the values of another stochastic process $S$, providing noisy information about the value of the firm, that can be thought, for example, as the price of an asset issued by the firm.\\ This observation process follows a diffusion of the type
\begin{equation}
\label{dSt}
\left \{
\begin{array}{rcl}
 d S_t &=& S_t \left[ \psi(V_t)   dt  + \nu(t) dW_{t} +  \delta(t)   d \bar W_t \right],\\ \quad S_0 & = & s_0,
\end{array}
\right.
\end{equation}
where  $\psi$ is locally bounded and Lipschitz, $ \nu$ and  $\delta$ are deterministic functions and $\bar W$ is a one-dimensional Brownian motion independent of $W$.
Note that in this model the return on $S$ is a (nonlinear) function of $V$.

Finally, following a structural approach, we define the default of the company as
\beq
\tau := \inf \left\{t \ge 0 : V_t \le a  \right\},
\eeq
where as usual $\inf \emptyset = + \infty$ and for a given constant parameter $a \in \mathbb R, 0 < a < v_0$.

In numerical examples we will consider models where $V_t \in (0,+\infty)$ (eventually by stopping the process $V$ at the default time $\tau$ by considering the process ${(V_{t \land \tau})}_t$).

We will deal with two different filtrations, representing different levels of information available to agents in the market and we suppose that they satisfy the usual hypotheses: a filtered probability space $(\Omega,\mathcal F, {(\mathcal F_t)}_{t \ge 0}, \mathbb P)$ satisfies the usual hypotheses if $\mathcal{F}_0$ contains all $\mathbb{P}$-null sets and if the filtration is right-continuous.

The first and basic information set is the ``default-free'' filtration, the one ge\-ne\-ra\-ted by the observation process $S$, which we will denote, for each $t \ge 0$,
$$
\mathcal F_t^S := \sigma ( S_s, 0 \le s \le t )
$$
and the second one is
the full information filtration ${(\mathcal G_t)}_{t \ge 0}$, i.e., the information available for example to a small number of stock holders of the company, who have access to $S$ and $V$ at each time $t$. In our case, the full information filtration is the one generated by the stochastic pair process $(W, \bar W)$.
In conclusion we have
$$
\mathcal F_t^S \subsetneq \mathcal G_t, \quad \forall \; t \ge 0,
$$
and we observe that the following \emph{immersion} property holds (see for example Coculescu, Geman and Jeanblanc, 2008 for an analogous analysis):
\begin{Lem}
Any ${(\mathcal F_t^S)}_t$-local martingale is a ${(\mathcal G_t)}_t$-local martingale.
We will say, then, that filtration ${(\mathcal F_t^S)}_t$ is \emph{immersed}  in the full filtration ${(\mathcal G_t)}_t$.
\end{Lem}
For a fixed $s \geq 0$ we observe the process $S$  from $0$ to $s$ and we suppose that the default $\tau$ occurs after $s$, so that,   $$ \tau = \inf\{t>s: \ V_t \leq a \}.$$
Suppose now that a finite time horizon $T$ is fixed. For a given $t$, such that $ s<t <T$, our aim is to compute the following key quantity
\beq
{1 \! \! 1}_{\{\tau>s\}} \mathbb P \left(\inf_{s \le u \le t} V_u > a \Big\vert \mathcal F_s^{S} \right) .
\eeq
We will see in the following Section \ref{NumRes} how this quantity plays a fundamental role (if computed under a pricing measure) in the computation of credit spreads for zero coupon bonds.

\subsection{Reduction to a non\-li\-ne\-ar filtering problem}
Using the law of iterated conditional expectations, the  Markov property of $V$ and the independence between $W$ and $\bar{W}$,  we find, for each $(s,t) \in \mathbb R^+ \times \mathbb R^+, s < t$,
\begin{eqnarray} \label{Pfiltrata}
\displaystyle \mathbb P \left(\inf_{s \le u \le t} V_u > a \Big\vert \mathcal F_s^{S} \right)
&=& \displaystyle  \mathbb E \left[ \mathbb P \left( \inf_{s \le u \le t} V_u >a \Big\vert \mathcal G_s \right) \Big\vert \mathcal F_s^{S} \right] \nonumber \\
& = & \displaystyle \mathbb E \left[ \mathbb P \left( \inf_{s \le u \le t} V_u >a \Big\vert V_s \right) \Big\vert \mathcal F_s^{S} \right] \nonumber \\
& = & \displaystyle \mathbb E \left[ F(s,t,V_s) \vert \mathcal F_s^{S} \right],   \qquad  \mathbb{P}{\rm -a.s.}
\end{eqnarray}
where, for every $x \in \mathbb{R}$,
\begin{equation} \label{DefinitionDeF}
 F(s,t,x) := \mathbb P \left( \inf_{s \le u \le t} V_u >a \Big\vert V_s =x \right).
 \end{equation}
Finally,
\bi
\item[(\textbf{P1})] \emph{if} we compute $F(s,t,x)$ for every $x \in \mathbb R$, which is now a conditional survival probability given the full filtration, and
\item[(\textbf{P2})] \emph{if} we obtain the \emph{filter distribution} at time $s$, i.e., the conditional distribution of $V_s$ given $\mathcal F_s^{S}$, $\Pi_{V_s \vert  \mathcal F_s^{S}}$,
\ei
then we are done, since it suffices to compute the integral
\begin{eqnarray*}
\mathbb E \left[ F(s,t,V_s) \vert \mathcal F_s^{S} \right] & = &
\int_{-\infty}^{\infty} F(s,t,x) \Pi_{V_s \vert \mathcal F_s^{S}} (\textrm d x) \\
& = &  \int_{a}^{\infty} F(s,t,x) \Pi_{V_s \vert \mathcal F_s^{S}} (\textrm d x).
\end{eqnarray*}
It remains to solve the two ``intermediate problems'' (\textbf{P1}) and (\textbf{P2}). Let us consider first problem (\textbf{P2}).

\section{Approximation of the filter by optimal quantization}\label{ApproxFilter}
Let us  recall first some facts about  optimal vector  quantization.
\subsection{A brief overview  on optimal quantization}
Consider an $\mathbb{R}^d$-valued random variable $X$ defined  on a probability space $(\Omega,\mathcal{A},\mathbb{P})$ with finite $r$-th moment and probability distribution $\mathbb P_X$. Quantizing $X$ on a given grid $\Gamma =\{ x^1,\cdots,x^N \}$ consists in projecting $X$ on the grid $\Gamma$ following the closest neighbor rule. The induced  mean  $L^r$-error
$$ \Vert X - {\rm Proj}_{\Gamma}(X) \Vert_r  =  \Vert \min_{1 \leq i \leq N} \vert X - x^i \vert \Vert_r$$
is called the {\em $L^r$-mean quantization error} and the projection of $X$ on $\Gamma$, ${\rm Proj}_{\Gamma}(X)$, is called {\em the quantization of $X$}.  As a function of the grid $\Gamma$ the $L^r$-mean quantization error is continuous and reaches a minimum over all the grids with size at most $N$.  A grid $\Gamma^{\star}$ minimizing the $L^r$-mean quantization error over all the grids with size at most $N$  is called an {\em $L^r$-optimal quantizer}.

Moreover, the $L^r$-mean quantization error goes to $0$ as the grid size $N \rightarrow +\infty$ and the convergence rate is ruled by Zador theorem:
$$  \min_{\Gamma,\ \vert \Gamma \vert = N} \Vert X - {\rm Proj}_{\Gamma}(X) \Vert_r = Q_r(\mathbb{P}_X) N^{-1/d} + o \big(N^{-1/d} \big) $$
where $Q_r(\mathbb P_X)$ is a nonnegative constant.  We shall say no more about the basic results on optimal vector quantization. For a complete background on this field we refer to Graf and Luschgy (2000).

The first application of optimal quantization method to numerical probability appears in Pag\`{e}s (1997). It consists in estimating $ \mathbb{E} f(X) $ (it may also be a conditional expectation) by
\begin{equation} \label{QuantProcedureEstim}
\mathbb{E} f \big({\rm Proj}_{\Gamma^{\star}}(X)  \big) = \sum_{i=1}^N f(x^{\star,i})  \ p_i
\end{equation}
where  $\Gamma^{\star} =\{x^{\star,1}, \cdots,x^{\star,N} \}$ is an $L^r$-optimal grid for $X$  and $p_i =\mathbb{P}\big( {\rm Proj}_{\Gamma^{\star}}(X) = x^{\star,i}\big).$  The induced quantization error estimate depends on the regularity of the function $f$.
\vskip 0.2cm
\noindent $\bullet$ If $f: \mathbb{R}^d \mapsto \mathbb{R}$ is Lipschitz continuous and $r \ge 2$, then
\begin{eqnarray*}
\vert \mathbb{E}f(X) -   \mathbb{E}  f \big({\rm Proj}_{\Gamma^{\star}}(X)  \big) \vert  & \leq & \mathbb{E} \vert  f(X) -  f \big({\rm Proj}_{\Gamma^{\star}}(X)  \big) \vert \\
&  \leq & [f]_{{\rm Lip}}  \Vert X -  {\rm Proj}_{\Gamma^{\star}}(X)   \Vert_1  \\
& \leq  &  [f]_{{\rm Lip}}  \Vert X -  {\rm Proj}_{\Gamma^{\star}}(X)   \Vert_2 .
\end{eqnarray*}

\noindent $\bullet$  If the derivative $Df$ of $f$ is Lipschitz and $r \ge 2$, then, for any optimal grid $\Gamma^{\star}$, we have
\begin{equation*}
\vert \mathbb{E}f(X) -  \mathbb{E} f \big({\rm Proj}_{\Gamma^{\star}}(X)  \big)   \vert \leq  [Df]_{{\rm Lip}}  \Vert X - {\rm Proj}_{\Gamma^{\star}}(X)  \Vert_2^2 .
\end{equation*}

How to compute numerically quadratic optimal quantizers  or $L^r$-optimal (or stationary) quantizers in general, the associated weights and $L^r$-mean quantization errors  is an important issue from the numerical point of view.  Several algorithms are used in practice. In the one dimensional framework, the $L^r$-optimal quantizers are unique up to the grid size as soon as the density of $X$ is strictly log-concave. In this case the Newton algorithm is a commonly used algorithm to carry out the $L^r$-optimal quantizers when closed or semi-closed formulas are available  for the gradient and the hessian matrix.

When the dimension $d$ is greater than $2$, on the contrary, the $L^r$-optimal grids are not uniquely determined and all  $L^r$-optimal quantizers search algorithms are based on zero search  recursive procedures like Lloyd's I algorithms (or generalized Lloyd's I algorithms which are the natural extension of the quadratic case), the Competitive Learning Vector Quantization (CLVQ) algorithm (see Gersho and Gray, 1992), stochastic algorithms (see Pag\`{e}s, 2008, and Pag\`{e}s and Printems, 2003), etc.  From now on we consider quadratic optimal quantizers.

\subsection{Estimation of the filter}\label{filtering}
We focus now on problem (\textbf{P2}) and we present here a solution based on optimal quantization, as suggested in Pag\`{e}s and Pham (2005).\\ For an overview on nonlinear filtering problems in interest rate and credit risk mo\-de\-ls we refer to Frey and Runggaldier (2009) and references therein and, focusing on filtering theory in credit risk, we also have to mention the seminal papers Kusuoka (1999) and Nakagawa (2001).\\

We consider a general \emph{discrete time} setting, in which we recall the relevant formulas and the desired approximation of the filter (see, e.g., Pag\`{e}s and Pham, 2005 and Pham, Runggaldier and Sellami, 2005, for a detailed background).
We introduce a probability space $(\Omega, \mathcal A,\mathbb P)$ (notice that $\mathbb P$ is not the same measure we considered in Section \ref{S.2}, but for simplicity we will use the same notation) and we suppose that:
\begin{itemize}
\item \emph{the signal process} ${(X_k)}_{k \in \mathbb N}$ is a  finite-state Markov chain taking values in the space $E$, with known probability transition, from time $k-1$ to time $k$, $P_k(x_{k-1},\textrm dx_k), k \ge 1$, and given initial law $\mu$, and
\item \emph{the observation process} is an $\mathbb R^q$-valued process ${(Y_k)}_{k \in \mathbb N}$ such that $Y_0=y_{0}$ and the pair ${(X_k,Y_k)}_{k \in \mathbb N}$ is a Markov chain.
\end{itemize}
Furthermore, we suppose that for all $k \ge 1$
\begin{itemize}
\item[(\textbf{H})] the law of $Y_k$ conditional on $(X_{k-1}, Y_{k-1}, X_k)$ admits a density
    $$
    y_k \mapsto g_k(X_{k-1}, Y_{k-1}, X_k, y_k),
    $$
\end{itemize}
so that the probability transition of the Markov chain ${(X_k,Y_k)}_{k \in \mathbb N}$ is given by\\ $P_k(x_{k-1}, \textrm d x_k) g_k(x_{k-1},y_{k-1},x_k, y_k) \textrm d y_k$, with initial law $\mu(\textrm d x_0) \delta_0 (\textrm d y_0)$.\\

In this discrete time setting we are interested in computing conditional expectations of the form
$$
\Pi_{Y,n}f := \mathbb E \left[ f(X_n) \vert Y_1, \dots, Y_n \right],
$$
for suitable functions $f$ defined on $E$, i.e., we are interested in computing at some time $n$ the law $\Pi_{Y,n}$ of $X_n$ given the past observation $Y=(Y_1,\dots,Y_n)$. Having fixed the observation $Y=(Y_1, \dots,Y_n) = (y_1,\dots,y_n)=:y$ we will write $\Pi_{y,n}$ instead of $\Pi_{Y,n}$.

It is evident that in the case when the state space of the signal consists of a finite number of points, the filter is characterized by a finite-dimensional vector: if for e\-xam\-ple each $X_k$ takes values in a set $\{ x_k^1, \dots,x_k^{N_k} \}$ (as in the  case where we quantize a process X at discrete times $t_{k}, k=0, \cdots,n$ with  grids of size $N_{k}$), then the discrete time filter distribution will be fully determined by the $N_k$-vector with components
$$\Pi_{Y,k}^i = \mathbb P \left( X_k = x_k^i \vert Y_1,\dots,Y_k \right), \quad i=1, \dots, N_k.
$$
It is for this reason that, following Pag\`{e}s and Pham (2005), we apply optimal quantization results in order to obtain a spatial discretization, on a grid $\Gamma_k=\{ x_k^1, \dots,x_k^{N_k} \}$, of the state $X_k,k=0, \dots,n$, and to characterize the filter distribution by means of the finite number of points $\{x_0,x_1^1,\dots,x_1^{N_1},x_2^1,\dots,x_2^{N_2},\dots,x_n^1,\dots, x_n^{N_n}\}$ making up the grids ${(\Gamma_k)}_k$.

In what follows we recall the basic recursive filtering equation, that we will use in our numerics to approximate the filter.
By applying the Markov property of $X$ and $(X,Y)$ and Bayes' formula, we find:
\begin{equation}  \label{EqFilterPi}
\Pi_{y,n} f = \frac{\pi_{y,n}f}{\pi_{y,n} 1\!\!1},
\end{equation}
where $\pi_{y,n}$ is the \emph{un-normalized filter}, defined by
\begin{equation}\label{filtroNN}
\pi_{y,n} f = \int \cdots \int f(x_n) \mu(\textrm d x_0) \prod_{k=1}^n g_k(x_{k-1},y_{k-1},x_k,y_k)P_k(x_{k-1},\textrm d x_k).
\end{equation}
Equivalently, we recall the following recursive formula, that can be directly obtained as well by applying Bayes' formula and the Markov property:
$$
\Pi_{y,k}(\textrm d x_k) 
\propto \int g_k(x_{k-1}, y_{k-1},x_k,y_k)P_k(x_{k-1},\textrm d x_k) \Pi_{y,k-1}(\textrm d x_{k-1}),
$$
where now $y$ in $\Pi_{y,k-1}$ represents the realization of the vector $(Y_1,\dots,Y_{k-1})$ and we do not have equality because we need to re-normalize.\\
Now for any $k \in \{1,\cdots,n \}$  note that
$$ \pi_{y,k} f= \mathbb{E} \Big( f(X_k) \prod_{i=1}^{k}  g_i(X_{i-1},y_{i-1},X_{i},y_i) \Big).$$
Therefore, introducing the natural filtration of $X$, $ {(\mathcal F_k^{{X}})}_{k \in \mathbb N}$, we have
\begin{eqnarray}
\pi_{y,k} f & =  &  \mathbb{E} \left(  \mathbb{E} \bigg( f(X_k) \prod_{i=1}^{k}  g_i(X_{i-1},y_{i-1},X_{i},y_i)  \vert \mathcal F ^X_{k-1} \bigg)  \right) \nonumber \\
&  = &   \mathbb{E} \left(  \mathbb{E} \left( f(X_k)  g_k(X_{k-1},y_{k-1},X_{k},y_k) \vert \mathcal F^X_{k-1} \right)   \prod_{i=1}^{k-1}  g_i(X_{i-1},y_{i-1},X_{i},y_i)   \right) \nonumber  \\
& = & \mathbb{E} \left(  H_{y,k} (f( X_{k-1}))    \prod_{i=1}^{k-1}  g_i(X_{i-1},y_{i-1},X_{i},y_i)  \right), \label{RecursiveForm}
\end{eqnarray}
where $H_{y,k}, k=1, \dots, n,$ is a  family of bounded transition kernels  defined on bounded measurable functions $f: E \rightarrow \mathbb R$ by:
\begin{eqnarray}
H_{y,k} f(x_{k-1}) &:=& \mathbb E \left[ f(X_k)g_k(x_{k-1},y_{k-1},X_k,y_k) \vert X_{k-1} = x_{k-1} \right]  \nonumber \\
& = & \int f(x_k) g_k(x_{k-1},y_{k-1},x_k,y_k)P_k(x_{k-1},\textrm d x_k), \label{KernelMatrix}
\end{eqnarray}
with $x_{k-1} \in E$. For every $x \in E$, furthermore, we have
$$
H_{y,0} f(x) := \pi_{y,0}f = \mathbb E \left[ f(X_0)\right]= \int f(x_0) \mu(\textrm d x_0).
$$
It follows, then, from (\ref{RecursiveForm}) that
\begin{equation}\label{filtroNNricorsiva}
\pi_{y,k}f = \pi_{y,k-1} H_{y,k}f, \quad k=1, \dots,n,
\end{equation}
so that we finally obtain the recursive expression
$$
\pi_{y,n} = H_{y,0} \circ H_{y,1} \circ \dots \circ H_{y,n}.
$$

In order to obtain the discrete time approximation of the desired filter $\Pi_{V_s \vert  \mathcal F_s^{S}}$, that is needed to solve (\textbf{P2}), we now fix a time discretization grid $t_0 =0 <\cdots<t_n=s$ in the interval $[0,s]$ and we apply the previous results by working with the corresponding quantized process $\widehat{V}$ (we identify $X$ with $V$ and  $Y$ with $S$).

Before focusing on the discrete time filter approximation, let us  make the following remark concerning the conditional law  of $S_t$ given $((V_{u})_{u \in [s,t]},S_s)$. This will be useful to be sure that in our case hypothesis (\textbf{H}) is verified.
\begin{Rem}  \label{PropLN}
{\rm Let $ s \leq t$. Using the form of the solution of the SDE (\ref{dSt})
$$ S_{t }  = S_s \exp \left(   \int_{s}^t   \Big(  \psi (V_{u})  - \frac{1}{2}({\nu^2(u)}+{\delta^2(u)}) \Big) du + \int_s^t \nu(u) d W_u  + \int_s^t \delta(u) d \bar{W}_u  \right),$$
we notice that
\beq  \label{EqCondLawSV}
\mathcal{L} \big(S_t   \vert (V_{u})_{ s \leq u \leq t},S_s \big) = {\rm LN} (m_{s,t}; \sigma_{s,t}^{2}),
\eeq
where
$$
m_{s,t} = \log(S_s) +  \int_{s}^t  \Big( \psi (V_u) - \frac{1}{2}({\nu^2(u)}+{\delta^2(u)}) - \nu(u) \frac{b(u,V_{u})}{\sigma(u,V_{u})} \Big)du   +  \int_s^t  \frac{\nu(u)}{\sigma(u,V_{u})} dV_{u}
$$
and
$$
\sigma_{s,t}^2  = \int_{s}^t {\delta^2(u)} du.
$$
${\rm LN}(m;\sigma^2)$ stands for the lognormal distribution with mean $m$ and variance $\sigma^2$.}
\end{Rem}
Now, suppose that we temporarily  have a time discretization grid  from $0$ to $t$: $u_{0}=0<u_{1}<\cdots<u_{m}=t$. For $m$ large enough we can estimate the mean  and the variance appearing in Equation (\ref{EqCondLawSV}) by using an Euler Scheme. When the estimations of  the mean  $m_{s,t}$ and  variance $\sigma^{2}_{s,t}$ between two discretization steps are respectively denoted by $m_{k}$ and $\sigma^{2}_{k}$ we have:
\beq  \label{EqCondLDisc}
\mathcal{L} \big(S_k   \vert  V_{k-1},S_{k-1},V_{k} \big) = {\rm LN} (m_{k}; \sigma_{k}^{2})
\eeq
with
\begin{eqnarray*}
m_{k} &=& \log S_{k-1} + \left( \psi(V_{k-1}) - \frac{1}{2} \left({\nu^2(u_{k-1})} + {\delta^2(u_{k-1})}\right)  - \nu(u_{k-1}) \frac{b(u_{k-1},V_{k-1})}{\sigma(u_{k-1},V_{k-1})} \right) \Delta_{k} \\
& & + \frac{\nu(u_{k-1})}{\sigma(u_{k-1},V_{k-1})} \Delta V_{k};
\end{eqnarray*}
$$
\sigma_{k}^{2} = {\delta^2(u_{k-1})} \Delta_{k},
$$
where   $ S_{k}:= S_{u_{k}};  \ V_{k} := V_{u_{k}};  \  \Delta V_{k} = V_{k} - V_{k-1}; \ \Delta_{k} = u_{k}-u_{k-1.}$
Recall that the density of the lognormal distribution with mean $m$ and variance $\sigma^2$, ${\rm LN}(m;\sigma^2)$, reads
$$  f(x;m,\sigma^2) = \frac{1}{\sigma x \sqrt{2 \pi }} e^{-\frac{1}{2 \sigma^2} (\log x - m)^2}.$$
We finally notice, then, that the density of $S_{k}$ given $(V_{k-1},S_{k-1},V_k)$ only depends on $(V_{k-1},S_{k-1},V_k)$.

\begin{Rem} {\rm
In the specific case where
\begin{equation*}
\left \{ \begin{array}{ll}
dV_t = \mu V_t  dt + \sigma V_t dW_t, & V_0 = v_0,    \\
dS_t = \mu S_t  dt + \sigma S_t dW_t + \delta S_t d\bar{W}_t,   & S_0 = s_0,  \\
\end{array}  \right.
\end{equation*}
we directly deduce from Remark \ref{PropLN} that for every $s \leq t$
\begin{equation*}
\mathcal{L} \big(S_t   \vert (V_{u})_{ s \leq u \leq t},S_s \big) = {\rm LN} \left( \log \Big(\frac{S_sV_t}{V_s} \Big) - \frac{1}{2} \delta^2(t-s); \delta^2(t-s) \right).
\end{equation*}
}
\end{Rem}

\vskip 0.5cm
\noindent $\rhd$  \textbf{Estimation of the filter.} The method is already studied in Pag\`{e}s and Printems (2003) and consists first in quantizing for every time step $t_k$ the random variable $V_k$ by considering
\begin{equation}
 \widehat{V}_k  = {\rm Proj}_{\Gamma_k}(V_k), \quad k=0,\cdots,n,
\end{equation}
where for every $k$, $\Gamma_k$ is a grid of $N_k$ points $v_k^i, \ i=1,\cdots,N_k$ to be optimally chosen and where ${\rm Proj}_{\Gamma_k}$ denotes the closest neighbor projection on the grid $\Gamma_k$.\\
Owing to Equation $(\ref{filtroNNricorsiva})$ our aim is to estimate the filter using an approximation of the probability transition $P_k(v_{k-1},dv_{k})$ of $V_{k}$ given $V_{k-1}$. These probability transitions are approximated by the probability transition matrix  $\hat{p}_k := (\hat{p}_k^{ij})$ of $\widehat{V}_{k}$ given   $\widehat{V}_{k-1}$:
\begin{equation} \label{transition_weights}
\hat{p}_k^{ij} = \mathbb{P} (\widehat{V}_k = v_k^j  \vert \widehat{V}_{k-1} = v_{k-1}^i), \ i=1,\cdots,N_{k-1}, \ j=1,\cdots,N_k.
\end{equation}
Then, following Equation $(\ref{KernelMatrix})$, by fixing the observation  $y:=(y_0, \cdots,y_{k})$ (for notational simplicity we indicate, as previously in the general case, the observations of $S$ by $y$), we estimate the transition kernel matrix $H_{y,k}$ by the quantized transition kernel $\widehat{H}_{y,k}$ given by
$$
\widehat{H}_{y,k}  =  \sum_{j=1}^{N_k} \widehat{H}_{y,k}^{ij} \delta_{v_{k-1}^i}, \quad i=1,\cdots,N_{k-1}, \quad k=1,\cdots,n,
$$
where
$$
\widehat{H}_{y,k}^{ij}  =   g_k(v_{k-1}^i,y_{k-1},v_k^j,y_k)  \hat{p}_k^{ij}, \quad i=1,\cdots,N_{k-1}, \ j=1,\cdots,N_k
$$
and where the $v_k^j$'s, $\ j=1,\cdots,N_k$ are the (quadratic) optimal quantizers of  $V_k$. We also estimate $H_{y,0}$ by
$$ \widehat{H}_{y,0} =  \sum_{i=1}^{N_0} \mathbb{P}(\widehat{V}_0 = v_0^i) \ \delta_{v_0^i}.$$
This leads to the following forward induction to estimate  $\pi_{y,n}$:
\begin{equation}
\widehat{\pi}_{y,0} = \widehat{H}_{y,0}, \qquad   \widehat{\pi}_{y,k} = \widehat{\pi}_{y,k-1} \widehat{H}_{y,k}, \quad \qquad  k=1,\cdots,n,
\end{equation}
or, equivalently,
\begin{equation*}
\left \{ \begin{array}{ll}
\widehat{\pi}_{y,0} = \widehat{H}_{y,0}  \\
\widehat{\pi}_{y,k} = \Big( \sum_{i=1}^{N_{k-1}}  \widehat{H}_{y,k}^{ij} \widehat{\pi}_{y,k-1}^i  \Big)_{j=1,\cdots,N_k}, \quad k=1,\cdots,n.
\end{array}  \right.
\end{equation*}
\noindent Finally, we estimate the desired filter 
at time $t_n=s$ by
\begin{equation} \label{OptimalFilter}
\widehat{\Pi}_{y,n} f = \frac{\widehat{\pi}_{y,n}f}{\widehat{\pi}_{y,n} 1\!\!1 }.
\end{equation}
\vskip 0.5cm
\noindent $\rhd$  \textbf{Estimation of the conditional survival probability.}
Owing to Equation $(\ref{Pfiltrata})$  we use  optimal quantization to estimate  the  $\mathbb{P} \big(\inf_{s \leq u \leq t} V_{u} >a \vert \mathcal{F}_s^{S} \big)$  on the set $\{ \tau>s \}$ by
\begin{equation}  \label{CondP}
 \sum_{i=1}^{N_n} F(s,t,v_n^i) \ \widehat{\Pi}_{y,n}^{i}
\end{equation}
where $v_n^i,i=1, \cdots, N_n$ is the quadratic optimal grid of the process $V$ at time $t_n=s$,  $\widehat{\Pi}_{y,n}^{i}$ is the $i$-th coordinate of the optimal filter $\widehat{\Pi}_{y,n}$ given in $(\ref{OptimalFilter})$ and, for every $i$, $F(s,t,v_n^i)$ is defined  as in $(\ref{DefinitionDeF})$.  Note that this last function has in general  no explicit expression. In such case, we will estimate it  by Monte Carlo  as specified in the next section.

\section{Approximation by Monte Carlo of survival probabilities under full information}\label{PsurvieMC}
The aim of this section is to solve problem (\textbf{P1}), i.e., to compute, for each pair of positive values $(s,t),s \le t \le T$,
\beq
\label{PrP1}
\mathbb P \left( \inf_{s \le u \le t} V_u >a \Big\vert V_s \right) = \mathbb E \left( {1 \! \! 1}_{\{\inf_{s \le u \le t} V_u >a\}} \vert V_s \right),
\eeq
where in our general setting the firm value $V$ follows \emph{a priori} a diffusion of the type (\ref{dVt}). Notice that in the specific case  where  $V$ is a geometric Brownian motion there exists a closed-formula, that we recall below.\\
If $$dV_t = \mu V_t dt + \sigma V_t dW_t, \quad V_0 =v_0,  $$
then
\beq  \label{PsurvieW}
\mathbb{P} \left(\inf_{s \leq u \leq t} V_u >a \vert V_s \right)  =  \Phi (h_1(V_s,t-s)) - \left( \frac{a}{V_s}\right)^{\sigma^{-2}(r-\sigma^2/2)} \Phi (h_2(V_s,t-s))
\eeq
where
\begin{eqnarray*}
h_1(x,u) &=& \frac{1}{\sigma \sqrt{u}} \left( \log \left(\frac{x}{a} \right) + \left(\mu - \frac{1}{2} \sigma^2\right) u \right), \\ h_2(x,u) &=& \frac{1}{\sigma \sqrt{u}} \left( \log \left(\frac{a}{x} \right) + \left(\mu - \frac{1}{2} \sigma^2\right) u \right)
\end{eqnarray*}
and where $\displaystyle  \Phi(x) = \frac{1}{\sqrt {2 \pi}} \int_{- \infty}^x e^{-u^2/2} \textrm d u$ is the cumulative distribution function of the standard Gaussian law (see for example Borodin and Salminen, 2002 or Revuz and Yor, 1999).\\

Since in general we cannot use directly the result in Equation (\ref{PsurvieW}), we have to resort to an approximation method. We  will adopt the ``regular Brownian bridge  method''.
However, note that many other methods can be used to estimate these probabilities, such as in Kahal\'{e} (2007), where the crossing probabilities are calculated via Schwartz distributions in the specific case of drifted Brownian motion and in Linetsky (2004), where the survival probabilities and hitting densities relative to the CIR, the CEV and to the OU diffusions are expressed as infinite series of exponential densities:
\begin{equation}\label{SerieLin}
\mathbb P_{v_0}(\tau > t) = \sum_{n=1}^{\infty} c_n e^{- \lambda_nt} , \quad t>0,
\end{equation}
where $0 < \lambda_1 < \lambda_2 < \dots < \lambda_n \rightarrow \infty$ as $n \rightarrow \infty$ and $(c_n)_n$ are explicitly given in terms of the solution of the Sturm-Liouville equation and the eigenvalues of the Sturm-Liouville problem.\\

In order to find an approximated solution to problem (\textbf{P1}) by means of the re\-gu\-lar Brownian bridge method, we focus on the interval $[s,t]$ and we discretize it by means of $u_0=s < u_1 < \dots < t=u_{N}$. We denote by $\bar V$ the continuous Euler scheme relative to $V$.

The regular Brownian bridge method is connected to the knowledge of the distribution of the minimum (or the maximum) of the continuous Euler scheme $\bar{V}$ relative to the process $V$ over the time interval $[s,t]$, given its values at the discrete time observation points $s=u_0 < u_1 < \dots < u_N=t$ (see, e.g., Pag\`{e}s, 2008).

\begin{Lem}
\begin{equation}
\mathcal{L} \Big(\min_{u \in [s,t]} \bar{V_u} \vert \bar{V}_{u_k} = v_k, k=0,\cdots,N \Big)   = \mathcal{L} \Big(\min_{k=0, \cdots,N-1} G_{v_{k},v_{k+1}}^{-1} (U_k) \Big)
\end{equation}
where $(U_k)_{k=0, \cdots,N-1}$ are $i.i.d$ random variables uniformly distributed over the unit interval and $G_{x,y}^{-1}$ is the inverse function of the conditional distribution function $G_{x,y}$, defined by
$$
G_{x,y}(u)  = \exp \Big( - \frac{2N}{(t-s) \sigma^2(x)}(u-x) (u-y) \Big) {1 \! \! 1}_{\{\min(x,y) \geq u \}}.
$$
Notice that we have omitted the dependence on time in  $\sigma$.
\end{Lem}
Recall that for a given time discretization grid  $u_k := s+ \frac{k(t-s)}{N}, k = 0,\cdots,N$, on the set $[s,t]$,  the continuous Euler scheme relative to the process $V$ is defined by
$$
\bar{V}_u  = \bar{V}_{\underline{u}}  + b(\underline{u},\bar{V}_{\underline{u}}) (u - \underline{u}) + \sigma(\underline{u},\bar{V}_{\underline{u}}) (W_{u} - W_{\underline{u}}), \quad  \bar{V}_s = v_s,
$$
with $\underline{u} = u_k $ if $u \in [u_k,u_{k+1}).$
We deduce from the previous lemma the following result.

\begin{Prop}   \label{PsurvieFull2}
$$ \mathbb P \left( \min_{s \le u \le t} \bar V_u > a \big \vert \bar V_s \right)  = \mathbb E \left(\prod_{k=0}^{N-1} G_{\bar V_{u_k},\bar V_{u_{k+1}}}(a) \big\vert  \bar V_s \right). $$
\end{Prop}

\begin{proof}[\bf{Proof}]
We have (recall that $\bar{V}_s = \bar{V}_{u_0}$)
\begin{eqnarray*}
\mathbb P \left( \min_{s \le u \le t} \bar V_u>a \big\vert \bar V_s \right) & = & \mathbb E \left( \mathbb P \left(  \min_{s \le u \le t} \bar V_u >a \big\vert \bar V_{u_k}, k=0, \dots, N \right) \big\vert  \bar V_s \right ) \nonumber \\
& = & \mathbb E \left( \textstyle 1\!\!1_{\{ \min_{0 \le k \le N} \bar V_{u_i} > a \}} \exp \Big(- \frac{2N}{t-s} \sum_{k=0}^{N-1} \frac{(\bar V_{u_k} - a)(\bar V_{u_{k+1}} - a)}{\sigma^2(\bar V_{u_k})}\Big) \big\vert  \bar V_s \right) \nonumber \\
& = & \mathbb E \left( \prod_{k=0}^{N-1} \textstyle 1\!\!1_{\{ \bar V_{u_k} > a; \bar V_{u_{k+1}} > a \}} \exp \Big(- \frac{2N}{t-s} \frac{(\bar V_{u_k} - a)(\bar V_{u_{k+1}} - a)}{\sigma^2(\bar  V_{u_k})} \Big) \big\vert  \bar V_s  \right) \nonumber \\
& = & \mathbb E \left( \prod_{k=0}^{N-1} G_{\bar V_{u_k},\bar V_{u_{k+1}}}(a) \big\vert  \bar V_s \right),
\end{eqnarray*}
which gives the announced result.
\end{proof}
\noindent By using Proposition \ref{PsurvieFull2}, we estimate the survival probability under full information
$$
\mathbb P \left( \inf_{s \le u \le t} \bar V_u >a \big\vert \bar V_s = v \right)
$$
by the following Monte-Carlo procedure:
\begin{itemize}
\item \emph{Time grid specification}. Fix $u_0=s < u_1 < \dots < t=u_{N}$, the set of $N+1$ points for the (discrete time) Euler scheme in the interval $[s,t]$;
\item \emph{Trajectories simulation}. Starting from $v$ and having fixed $\emph{M}$ (number of Monte Carlo  simulations), for $j=1, \dots, M$, simulate the discrete path  ${(\bar V_{u_k}^j)}_{k=0, \dots,N}$;
\item \emph{Computation of the survival probability}. For $j=1, \dots, M$, compute (recall that, for every $j$, $\bar V_{u_0}^j = v$)
\begin{equation} \label{DefinitionOfProd}
p_{s,t}^j(v;a) := \prod_{k=0}^{N-1} G_{\bar V_{u_k}^j,\bar V_{u_{k+1}}^j}(a).
\end{equation}
\item \emph{Monte Carlo procedure}. Finally, apply the
Monte-Carlo paradigm and get the following approximating value
\begin{equation} \label{EstimationByMC}
\mathbb P \left( \inf_{s \le u \le t} \bar V_u >a \Big\vert \bar V_s  = v \right)  \approx  \frac{\sum_{j=1}^{M} p_{s,t}^j(v;a) }{M}.
\end{equation}
\end{itemize}

As a consequence, combining formulas (\ref{CondP}) and (\ref{EstimationByMC}) leads to the following hybrid Monte Carlo - optimal quantization formula on the set $\{ \tau>s \}$
\begin{equation} \label{FinalFormulaSurvival}
\mathbb{P} \left(\inf_{s \leq u \leq t} V_{u} >a \Big\vert \mathcal{F}_s^{S} \right) \approx  \frac{1}{M} \sum_{j=1}^{M} \sum_{i=1}^{N_n}  p_{s,t}^{j}(v_n^i;a) \  \widehat{\Pi}_{y,n}^{i}
\end{equation}
where $ p_{s,t}^{j}(\cdot \; ;a) $ was introduced in  $(\ref{DefinitionOfProd})$.
\section{The error analysis}\label{Error}
We now focus on the error done by approximating  $\displaystyle \mathbb{P} \left(\inf_{s \leq u \leq t} V_{u} >a \Big\vert \mathcal{F}_s^{S} \right) $ by
$$
\frac{1}{M} \sum_{j=1}^{M} \sum_{i=0}^{N_n}  p_{s,t}^{j}(v_n^i;a) \  \widehat{\Pi}_{y,n}^{i}.
$$
We distinguish three types of error. The first error is induced by the approximation of the filter $\Pi_{y,n}$ appearing in Equation (\ref{EqFilterPi})  by $\hat{\Pi}_{y,n}$, defined in (\ref{OptimalFilter}). The second one is the error deriving from the  approximation of  $$
\mathbb P \left( \inf_{s \le u \le t}  V_u >a \big\vert V_s = v \right) \ \textrm{ by } \   \mathbb P \left( \inf_{s \le u \le t} \bar V_u >a \big\vert \bar V_s = v \right),
$$
where $\bar V$ is the (continuous) Euler scheme relative to the process $V$. The last one is the error arising from the approximation of the survival probability under full information by means of Monte Carlo simulation.

Note that in the case when Equation (\ref{dVt}) admits an explicit solution, as in the Black-Scholes model, there is no need to use an  Euler scheme, so that the second kind of error has no more to be taken into account.
\vskip 0.5cm
\noindent $\rhd$ {\it  \bf Filter approximation error}. In order to have some upper bound of the quantization error estimate of  $\Pi_{y,n} F(s,t,\cdot)$ let us make the following assumptions, as done in Pag\`{e}s and Pham (2005). Notice that they are stated in the more general setting of Section \ref{filtering}, so that we have to identify $V$ with $X$ and $S$ with $Y$.
\vskip 0.3cm
\noindent $\textbf{(A1)}$ \ The transition operators $P_k(x,dy)$ of $X_k$ given $X_{k-1}, \ k=1,\cdots,n$ are Lipschitz.
\vskip 0.3cm
Recall that a probability transition $P$ on $E$ is ${\rm C}$-Lipschitz  (with ${\rm C}>0$) if for any Lipschitz function $f$ on $E$ with ratio $[f]_{Lip}$, $Pf$ is Lipschitz with ratio $[Pf]_{Lip} \leq {\rm C} [f]_{Lip}$.  Then, one may define the Lipschitz ratio $[P]_{Lip}$ by
\begin{equation*}
[P]_{Lip} = \sup \Big\{ \frac{[Pf]_{Lip}}{[f]_{Lip}}, f  \textrm{ a  nonzero Lipschitz function }  \Big\} < +\infty .
\end{equation*}
If  the transition  operators  $P_k(x,dy),\ k=1,\cdots,n$   are Lipschitz, it follows  that
$$ [P]_{Lip} := \max_{k=1,\cdots,n}[P_k]_{Lip} <+\infty. $$
\vskip 0.3cm
\begin{itemize}
\item[\textbf{(A2)}]
    \bi
    \item [(i)] For every $k=1,\cdots,n$, the functions  $g_k$ (recall hypothesis (\textbf{H})) are bounded on $E \times \mathbb{R}^q \times E \times \mathbb{R}^q$ and we set
        $$ {\rm K}^n_g := \max_{k=1,\cdots,n} \Vert  g_{k}  \Vert_{\infty}. $$
    \item[(ii)]  For every $k=1,\cdots,n$, there  exist  two positive  functions $[g^1_{k}]_{Lip}$ and $[g^2_{k}]_{Lip}$ defined on $\mathbb{R}^q \times \mathbb{R}^q$ so that for every $x,x',\widehat{x},\widehat{x}' \in E$ and $y,y' \in \mathbb{R}^q$,
$$  \vert g_{k}(x,y,x',y') - g_{k}(\widehat{x},y,\widehat{x}',y') \vert  \leq [g^1_{k}]_{Lip}(y,y')\ \vert x-\widehat{x} \vert + [g^2_{k}]_{Lip}(y,y') \ \vert x'-\widehat{x}' \vert .$$
    \ei
\end{itemize}
The following result gives the error bound of the estimation of the filter (see Pag\`{e}s and Pham, 2005, for details of the proof).
\begin{Thm}  \label{ThmPhamPages}
Suppose that Assumptions $\textbf{(A1)}$ and $\textbf{(A2)}$ hold true. For e\-ve\-ry bounded Lipschitz function $F$ on $E$ and for every $n$-tuple of observations $y=(y_1,\cdots,y_n)$, we have for every $p \geq 1$,
\begin{equation}
\vert   \Pi_{y,n} F - \widehat{\Pi}_{y,n} F \vert   \leq  \frac{{\rm K}^n_g }{\phi_n(y) \vee  \hat{\phi}_n(y)}  \sum_{k=0}^n {\rm B}_k^n(F,y,p) \ \Vert X_k - \widehat{X}_k  \Vert_p
\end{equation}
with
$$ \phi_{n}(y) := \pi_{y,n} \mbox{\bf{1}},  \quad \widehat{\phi}_{n}(y) := \widehat{\pi}_{y,n} \mbox{\bf{1}} $$  and
\begin{eqnarray*}
{\rm B}_k^n(F,y,p) & := &  (2-\delta_{2,p}) [P]_{Lip}^{n-k} [F]_{Lip} + 2 \bigg( \frac{\Vert F\Vert_{\infty}}{K_g} \big([g_{k+1}^1]_{Lip}(y_k,y_{k+1}) +[g_{k}^2]_{Lip}(y_{k-1},y_{k}) \big) \\
& + &  (2-\delta_{2,p})  \frac{\Vert F \Vert_{\infty}}{K_g} \sum_{j=k+1}^n [P]_{Lip}^{j-k-1} \big([g_{j}^1]_{Lip}(y_{j-1},y_{j}) +[P]_{Lip} [g_{j}^2]_{Lip}(y_{j-1},y_{j}) \big) \bigg).
\end{eqnarray*}
(Convention: $g_0 = g_{n+1} \equiv 0$ and $\delta_{n,p}$ is the usual Kronecker symbol).
\end{Thm}
\vskip 0.5cm
\noindent $\rhd$ {\bf Error induced by the Euler scheme}.
We here refer to Gobet (1998), in which the author starts by investigating the case of a one-dimensional diffusion and to the successive related article Gobet (2000) for the multidimensional case. In the two papers the considered diffusion has homogeneous coefficients $b$ and $\sigma$. We start by recalling here some important convergence results we find therein, we will then adapt these results to our case.

Suppose that $X$ is a diffusion taking values in $\mathbb{R}$, with $X_0 =x$, and define  $\tau'$ as the first exit time from an open set $D \subset \mathbb R$:
$$
\tau' := \inf \left\{u \ge 0 : X_u \not \in D  \right\}.
$$
Let $\tau_c'$ denotes the exit time from the domain $D$ of the continuous Euler process $\bar{X}$.  In order to give the error bound in the approximation of $\mathbb{E}_x \left( {1 \! \! 1}_{\{ \tau'>t \}} f(X_t) \right)$ by  $\mathbb{E}_x \left( {1 \! \! 1}_{\{ \tau_c'>t \}} f(\bar X_t) \right)$  the following hypotheses are needed:
\begin{center}
\begin{itemize}
\item [(\textbf {H1})] $b$ is a $\mathcal C_b^{\infty}(\mathbb R,\mathbb R)$ function and $\sigma$ is in $\mathcal C_b^{\infty}(\mathbb R,\mathbb R)$,
\item [(\textbf {H2})] there exists $\sigma_0>0$ such that $\forall x \in \mathbb R, \sigma(x)^2 \ge \sigma_0^2  $ (\emph {uniform ellipticity}),
\item [(\textbf {H3})]  $\ \MP_x \left( \inf_{t \in [0,T]} X_t = a \right) =0$.
\end{itemize}
\end{center}
The following proposition states that, under hypothesis (\textbf {H3}), the approximation error goes to zero as the number of time discretization steps goes to infinity.
\begin{Prop}[\textbf{Convergence}]\label{1DconvGobet}
Suppose that $b$ and $\sigma$ are Lipschitz, \mbox{$D=(a,+\infty)$} and that (\textbf {H3})  holds. If $f \in \mathcal C_b^0(\bar D,\mathbb R)$ then,
$$
\lim_{N \rightarrow +\infty}  \Big\vert \mathbb E_x[ {1 \! \! 1}_{\{\tau_c' > T\}} f(\bar X_T) ] - \mathbb E_x[ {1 \! \! 1}_{\{\tau' > T\}} f(X_T) ] \Big\vert = 0.
$$
\end{Prop}
\noindent Note that in the homogeneous case, when $D=(a,+\infty)$, a sufficient condition in order for (\textbf{H3}) to hold is
\beq \label{CondC}
\sigma(X_0) \neq 0.
\eeq
On the other hand, the rate of convergence  is given by the following
\begin{Prop}[\textbf{Rate of convergence}]\label{PropErrCsqrtN}
Under hypotheses (\textbf{H1}) and (\textbf{H2}), if $f \in \mathcal C_b^1(\bar D, \mathbb R)$, then there exists an increasing function $K(T)$ such that
$$
\Big\vert \mathbb E_x[ {1 \! \! 1}_{\{\tau_c' > T\}} f(\bar X_T) ] - \mathbb E_x[ {1 \! \! 1}_{\{\tau' > T\}} f(X_T) ] \Big\vert \le \frac{1}{\sqrt N} K(T) {\vert \vert f \vert \vert}_D^{(1)},
$$
where ${\vert \vert f \vert \vert}_D^{(1)} = \sum_{j=0}^1 \sup_{x \in D} \vert f^{(j)}(x) \vert$.
\end{Prop}
\vspace{0.3 cm}
\begin{Rem}{\rm  One notes that the two previous propositions still hold when the diffusion coefficients are in-homogeneous, as in our general setting, by replacing Hypotheses (\textbf {H1}), (\textbf {H2}),  (\textbf {H3}) by  (\textbf {I}) and (\textbf {J}) :
\begin{center}
\begin{itemize}
\item [(\textbf{I})] $b$ and $\sigma$ are $\mathcal C_b^{\infty}$ functions with respect to both arguments $t$ and $v$, with uniformly bounded partial derivatives with respect to $v$,
\item [(\textbf{J})] $\sigma$ is uniformly elliptic, i.e., $\exists \alpha >0$ such that $\sigma^2(t,v) \ge \alpha, \forall (t,v) \in [0,T] \times \mathbb R$ and $\sigma(0, v_0) \neq 0$.
\end{itemize}
\end{center}
}
\end{Rem}

\vskip 0.5cm
\noindent $\rhd$ {\bf Error induced by Monte Carlo approximation}. This error  comes from the  estimation of  $ \mathbb P \left( \min_{s \le u \le t} \bar V_u > a \big \vert \bar V_s =v_s^i \right)  = \mathbb E \left(\prod_{k=0}^{N-1} G_{\bar V_{u_k},\bar V_{u_{k+1}}}(a) \big\vert  \bar V_s = v_s^i \right)$, for every $i=1,\cdots,N_s$,  by  $$ \frac{\sum_{j=1}^{M} p_{s,t}^j(v_s^i;a) }{M},$$
where  $p_{s,t}^j(\cdot \; ; a)$ is defined in (\ref{DefinitionOfProd}). We have for every $i=1,\cdots,N_s$,
\beq  \label{EqMCError}
\Big\Vert  \mathbb E \Big(\prod_{k=0}^{N-1} G_{\bar V_{u_k},\bar V_{u_{k+1}}}(a) \big\vert  \bar V_s = v_s^i \Big) -  \frac{\sum_{j=1}^{M} p_{s,t}^j(v_s^i;a) }{M} \Big\Vert_2   = \mathcal O \Big(\frac{1}{\sqrt{M}} \Big).
\eeq

\vskip 0.5cm
By adapting the previous results to our case, namely by identifying $V$ with $X$ and $S$ with $Y$, one deduces an error bound for the estimation of $\Pi_{y,n} F(s,t,\cdot)$ by $\widehat{\Pi}_{y,n} F_{\scriptscriptstyle MN}(s,t,x)$, where $F_{\scriptscriptstyle MN}(s,t,x)$ is a Monte Carlo estimation of $F(s,t,\cdot)$ of size $M$, based on a time discretization grid, between $s$ and $t$, of size $N$.

\begin{Thm} \label{ThmConvergence}  Suppose that the Assumptions of  Theorem \ref{ThmPhamPages} and Proposition \ref{PropErrCsqrtN} hold. Then
\begin{eqnarray*}
\vert   \Pi_{y,n} F(s,t,\cdot) - \widehat{\Pi}_{y,n} F_{\scriptscriptstyle MN}(s,t,\cdot)   \vert  &  \leq & \frac{{\rm K}^n_g }{\phi_n(y) \vee  \hat{\phi}_n(y)}  \sum_{k=0}^n {\rm B}_k^n(F(s,t,\cdot),y,p) \ \Vert V_k - \widehat{V}_k  \Vert_p  \\
&  + &  \mathcal O  \left( \frac{1}{\sqrt{N}} \right)  + \mathcal O  \left( \frac{1}{\sqrt{M}} \right),
\end{eqnarray*}
where $n$ is the dimension of the observation vector $y$, $N$ stands for the size of the time discretization grid for the Euler scheme from $s$ to $t$ and $M$ is the number of Monte Carlo trials. Furthermore, $K^n_g, \phi_n(y),  \hat{\phi}_n(y)$ and $B_k^n, k=0,\dots,n$, are introduced in Theorem \ref{ThmPhamPages}.
\end{Thm}

\begin{proof}[$\textbf{Proof}$]
We have
\begin{eqnarray*}
\vert   \Pi_{y,n} F(s,t,\cdot) - \widehat{\Pi}_{y,n} F_{\scriptscriptstyle MN}(s,t,\cdot)   \vert  & \leq & \vert  \Pi_{y,n} F(s,t,\cdot) - \widehat{\Pi}_{y,n} F(s,t,\cdot)    \vert  \\
& + & \vert  \widehat{\Pi}_{y,n} F(s,t,\cdot) - \widehat{\Pi}_{y,n} F_{\scriptscriptstyle MN}(s,t,\cdot)  \vert.
\end{eqnarray*}
The error bound of the first term on the right-hand side of the above inequality is given by Theorem \ref{ThmPhamPages}. As concerns the last term we have
\begin{eqnarray*}
\vert  \widehat{\Pi}_{y,n} F(s,t,\cdot) - \widehat{\Pi}_{y,n} F_{\scriptscriptstyle MN}(s,t,\cdot)  \vert  & = & \Big\vert  \sum_{i=1}^{N_s} \widehat{\Pi}_{y,n}^i(F(s,t,v_s^i) - F_{\scriptscriptstyle MN}(s,t,v_s^i))  \Big\vert  \\
& \leq &  \sup_{v \in  \mathbb{R}}  \vert  F(s,t,v) - F_{\scriptscriptstyle MN}(s,t,v)  \vert  \sum_{i=1}^{N_s} \widehat{\Pi}_{y,n}^i \\
&  = & \sup_{v \in  \mathbb{R}}  \vert  F(s,t,v) - F_{\scriptscriptstyle MN}(s,t,v)  \vert .
\end{eqnarray*}
On the other hand, we have for every $v \in \mathbb{R}$
\begin{eqnarray*}
\vert F(s,t,v) - F_{\scriptscriptstyle MN}(s,t,v)    \vert
& \le  &  \Big\vert \mathbb{P}_v(\tau>t) - \mathbb E_{v} \Big(\prod_{k=0}^{N-1} G_{\bar V_{u_k},\bar V_{u_{k+1}}}(a) \Big) \Big\vert \\
& + & \Big \Vert  \mathbb E_{v} \Big(\prod_{k=0}^{N-1} G_{\bar V_{u_k},\bar V_{u_{k+1}}}(a) \Big) -  \frac{\sum_{j=1}^{M} p_{s,t}^j(v;a) }{M} \Big \Vert_2.
\end{eqnarray*}
We then deduce from Proposition \ref{PropErrCsqrtN} and from Equation (\ref{EqMCError}) that
$$  \vert F(s,t,v) - F_{\scriptscriptstyle MN}(s,t,v)    \vert  \le   \mathcal O  \left( \frac{1}{\sqrt{N}} \right)  + \mathcal O  \left( \frac{1}{\sqrt{M}} \right), $$
which completes the proof.
\end{proof}

\begin{Rem} {\rm (About the hypotheses of Theorem \ref{ThmConvergence})
We consider the case when $V$ is a time homogeneous diffusion.

\noindent $\rhd$  It is straightforward, by using Remark  \ref{PropLN}, to prove that Assumption {\bf (A2)} is fulfilled.

\noindent  $\rhd$   If we suppose that the coefficients $b$ and  $\sigma$ of the diffusion $V$ are Lipschitz, we show, by using the Euler scheme relative to $V$, that Hypothesis {\bf (A1)} holds true.

\noindent $\rhd$   As concerns the Lipschitz property of the function $F(s,t,\cdot)$, it follows from Proposition 2.2.1 in Gobet (1998), in the case when the coefficients of the diffusion satisfy Hypotheses (\textbf{H1}) and (\textbf{H2}) and for $t > s$.
}
\end{Rem}

\section{Numerical results}\label{NumRes}
In the numerical experiments we deal with the estimation of the credit spread for zero coupon bonds. We fix $s$ and, given the observations of $S$ from $0$ to $s$, we estimate the spread curve for different maturities $t \; (t > s)$.

The credit spread is the difference in yield between a corporate bond  and a risk-less bond (treasury  bond) with the same characteristics. It can be seen as a measure of the riskiness relative to a corporate bond, with respect to a risk-free bond. If we suppose for simplicity  that the face value is equal to $1$ and the recovery rate is zero, the credit spread  under partial information from time $s$ to maturity $t$, $S(s,t)$, equals (see Bielecki and Rutkowski, 2004 and Coculescu, Geman and Jeanblanc, 2008)
$$ S(s,t)  =  - \frac{\log \big(\mathbb{Q}(\inf_{s < u \leq t} V_u >a \vert \mathcal F_s^{S} ) \big)}{t-s},$$
where $\mathbb Q$ is a martingale  measure equivalent to $\mathbb P$.  We suppose that the market is complete, so that $\mathbb Q$ is unique.

We will consider two models for the dynamics of the firm value $V$: the Black-Scholes one and a CEV (Constant Elasticity of Variance) model. In both cases we fix $s=1$ and we estimate the spreads, given the observed trajectory of $S$ from $0$ to $s$, for different maturities varying $0.1$ by $0.1$ from $1.1$ to $11$ (the time unit is expressed, e.g., in years). We set the number $n$ of discretization points over $[0,s]$ equal to $50$ and for every $k =1, \cdots,n$,  the quantization grid size $N_k$ is set to $60$, with $N_0 =1$.  The number of Euler discretization steps  $N$ equals $50$ for $t$ varying $0.1$ by $0.1$ from $1.1$ to $3.0$ and $N=100$  for $t$ varying $0.1$ by $0.1$ from $3.1$ to $11.0$. The number of Monte Carlo trials $M$ is set to $300000$. In the quantization phase we obtain the optimal grid by carrying out   $80$ Lloyd's I procedures.
\vskip 0.3cm
\noindent  $\rhd$ {\it The Black-Scholes model.} We consider the following model for the firm value's and the observed process' dynamics:
\begin{equation}
\label{DinBS}
\left \{
\begin{array}{rclr}
dV_t  &=&  V_t ( r dt  +  \sigma  dW_t), & V_0 = v_0,\\
d S_t & = & S_t (r dt + \sigma d W_t + \delta d \bar W_t), & S_0=v_0,
\end{array}
\right.
\end{equation}
with $r=0.03$, $\sigma = 0.05$, $v_0=86.3$. The barrier $a$ is fixed to $76$ and $\delta=0.1$. We have, in particular,
\beq   \label{EqReturnS}
\frac{dS_t}{S_t} = \frac{dV_t}{V_t} + \delta d \bar{W}_t,
\eeq
meaning that the return on $S$ is the return on $V$ affected by a noise.
\begin{Rem}{\rm
a) Since $V$ is not traded in the market, the return on $V$ is not necessarily equal to the interest rate $r$. We have set the return on $V$ to be equal to $r$ in order to have a good economic interpretation of the model.\\
b) Note that in this particular case all the hypotheses of Theorem \ref{ThmConvergence} are satisfied.
}
\end{Rem}
Numerical results are presented in Figures \ref{figureSBS} and \ref{figureVBS}. Figure \ref{figureSBS} is relative to the partial information case, so that three trajectories of the observed process $S$ and the corresponding credit spreads are depicted. Figure \ref{figureVBS}, on the contrary, concerns the full information case, namely we suppose that we can observe different trajectories of $V$, that are drawn on the left, and we compute the corresponding credit spreads, that are shown on the right.

Remark that in this setting we deduce from  (\ref{DinBS}) that
$$  S_t = V_t e^{-\frac{1}{2} \delta^2 t + \delta \bar{W}_t}. $$
The correlation coefficient is then given for every $t$ by
$$
\rho(t):= \sqrt{\frac{e^{\sigma^2 t }-1}{e^{(\sigma^2 +\delta^2)t}-1}},
$$
meaning that the firm value $V$ is positively correlated to the observation process $S$. Observe that when $\sigma<\delta$, $\rho(t)$  is a strictly decreasing function and goes to $0$ as $t$ goes to infinity. This tell us that the a posteriori information on $V$ given  $S$ decreases as the maturity $t$ increases. This is what we observe in the spreads curves from Figures \ref{figureSBS} and \ref{figureVBS}, since for large maturities the spreads values almost coincide  for analogous  trajectories.

\begin{figure}[htpb]
  \includegraphics[width=8.3cm,height=6.0cm]{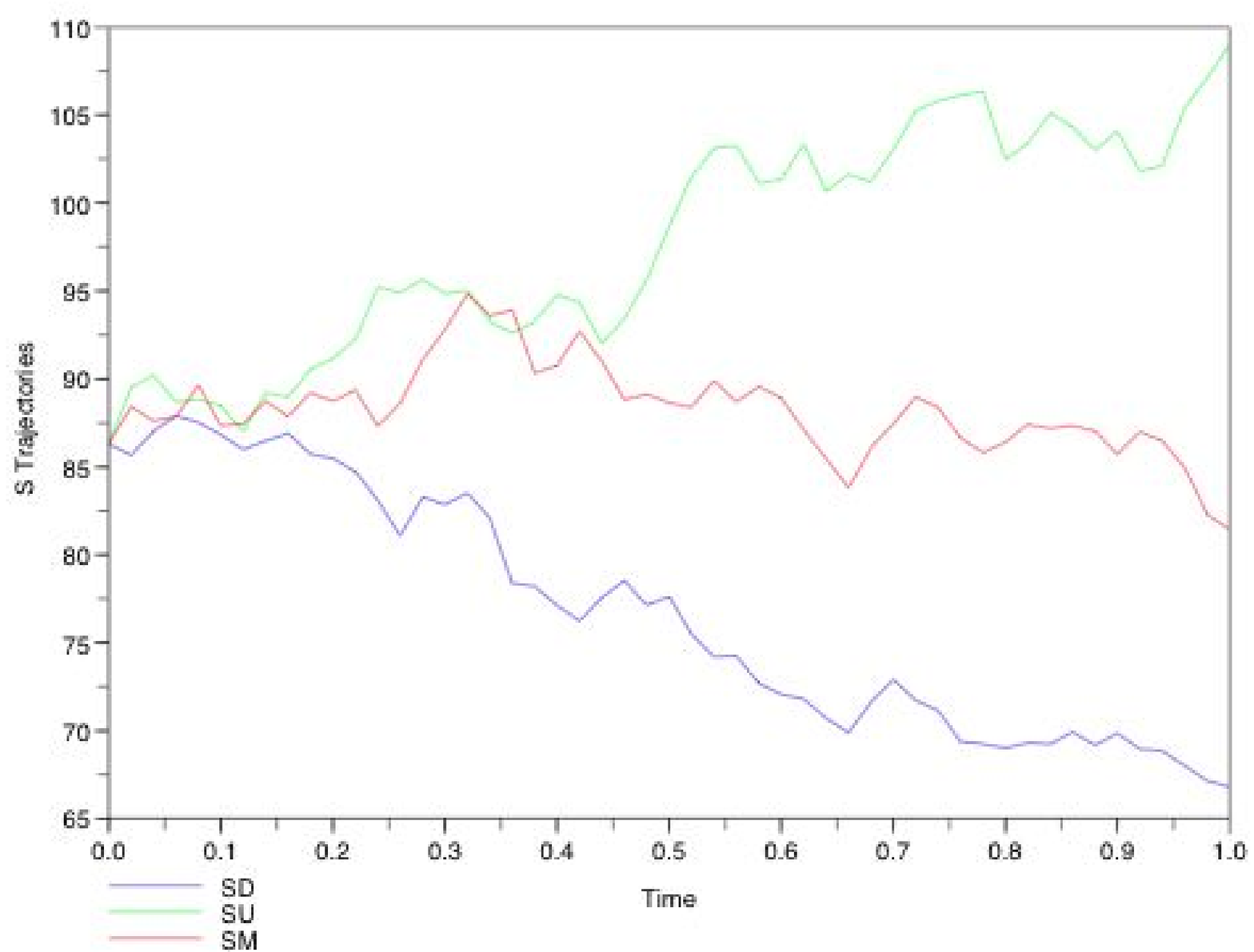}
  \hfill \includegraphics[width=8.3cm,height=6.0cm]{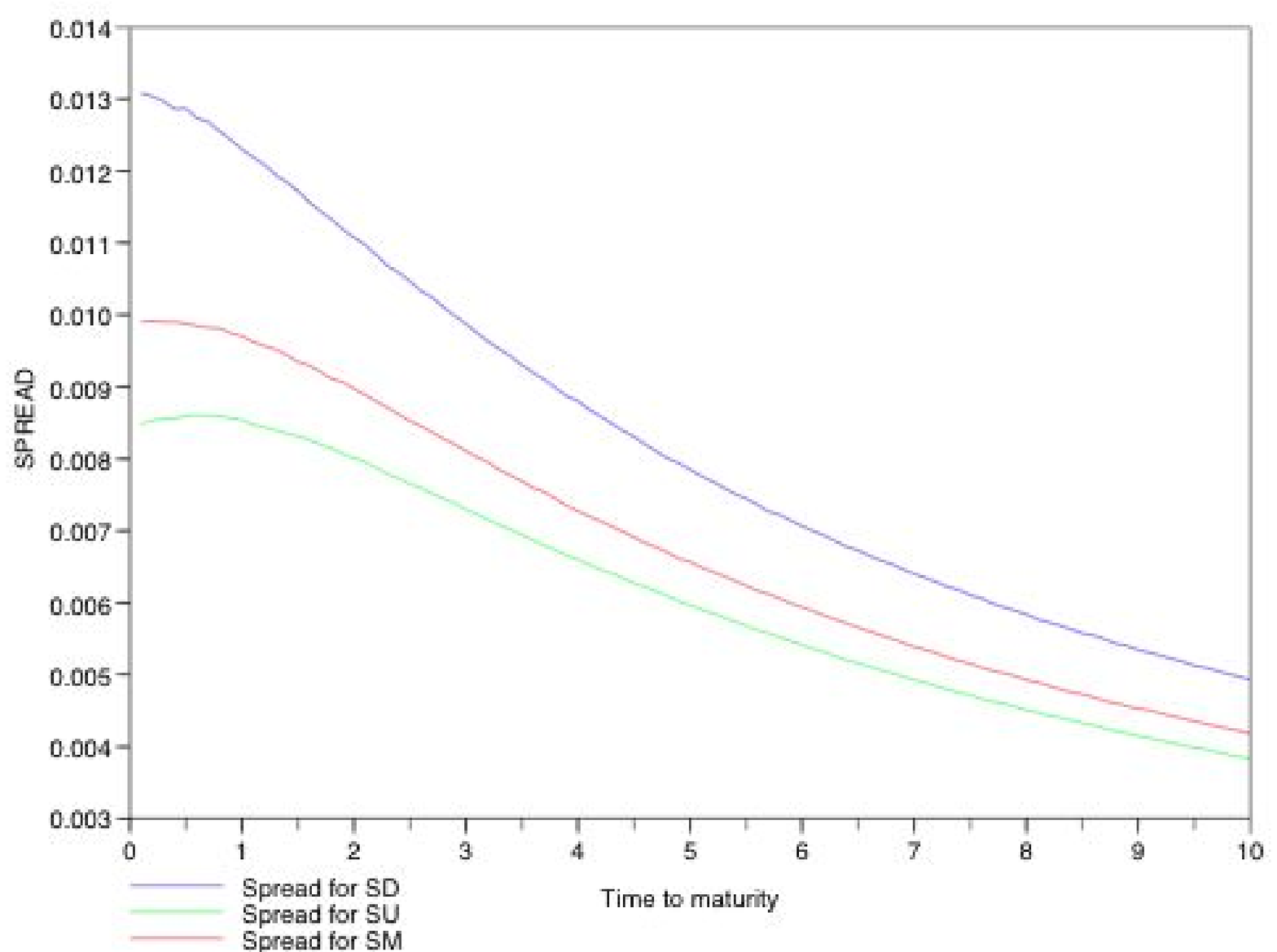}
  \caption{\small{Three trajectories of the observed process $S$ (on the left) and the corresponding spreads  (on the right).}}
\label{figureSBS}
  \end{figure}

  \begin{figure}[htpb]
  \includegraphics[width=8.3cm,height=6.0cm]{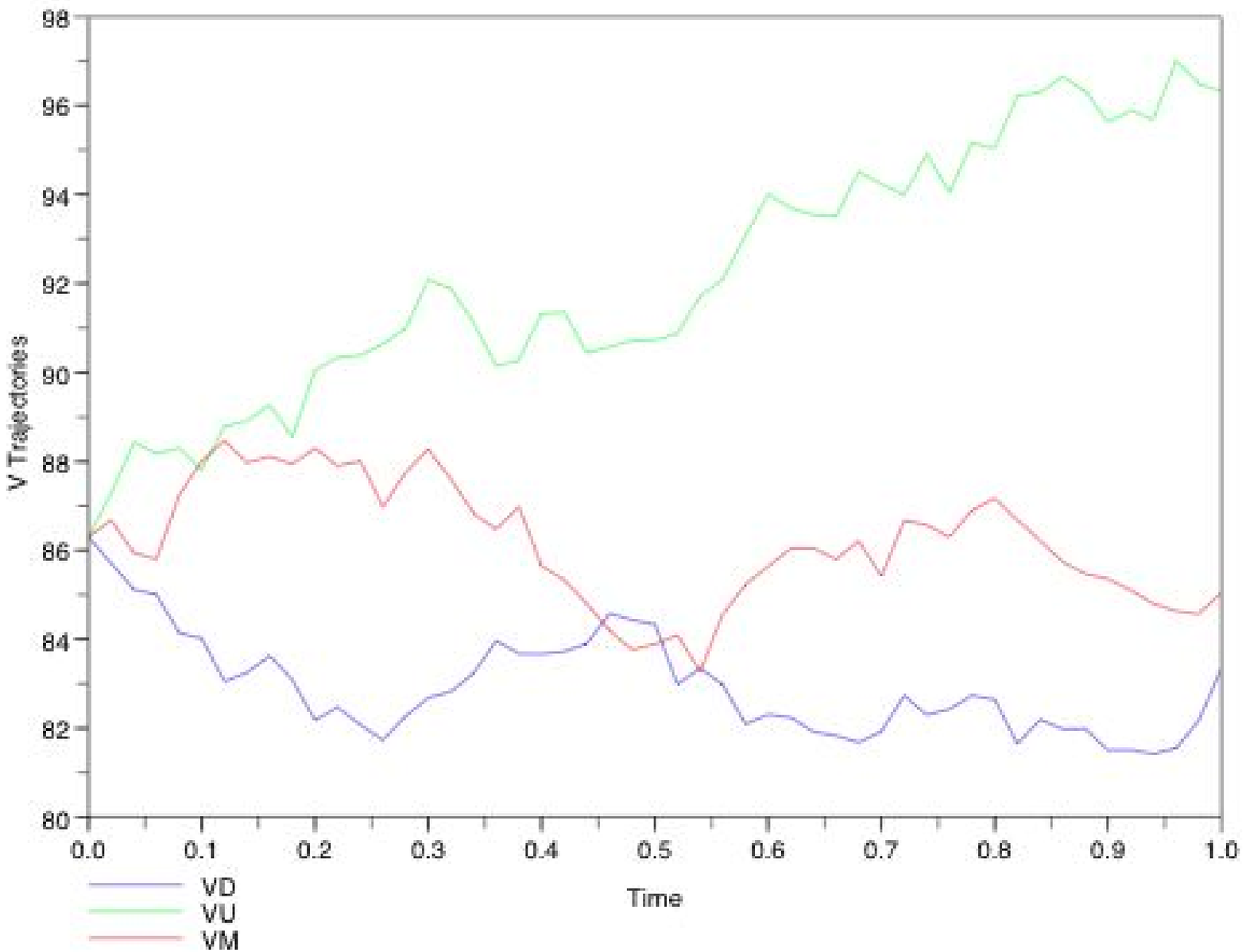}
  \hfill \includegraphics[width=8.3cm,height=6.0cm]{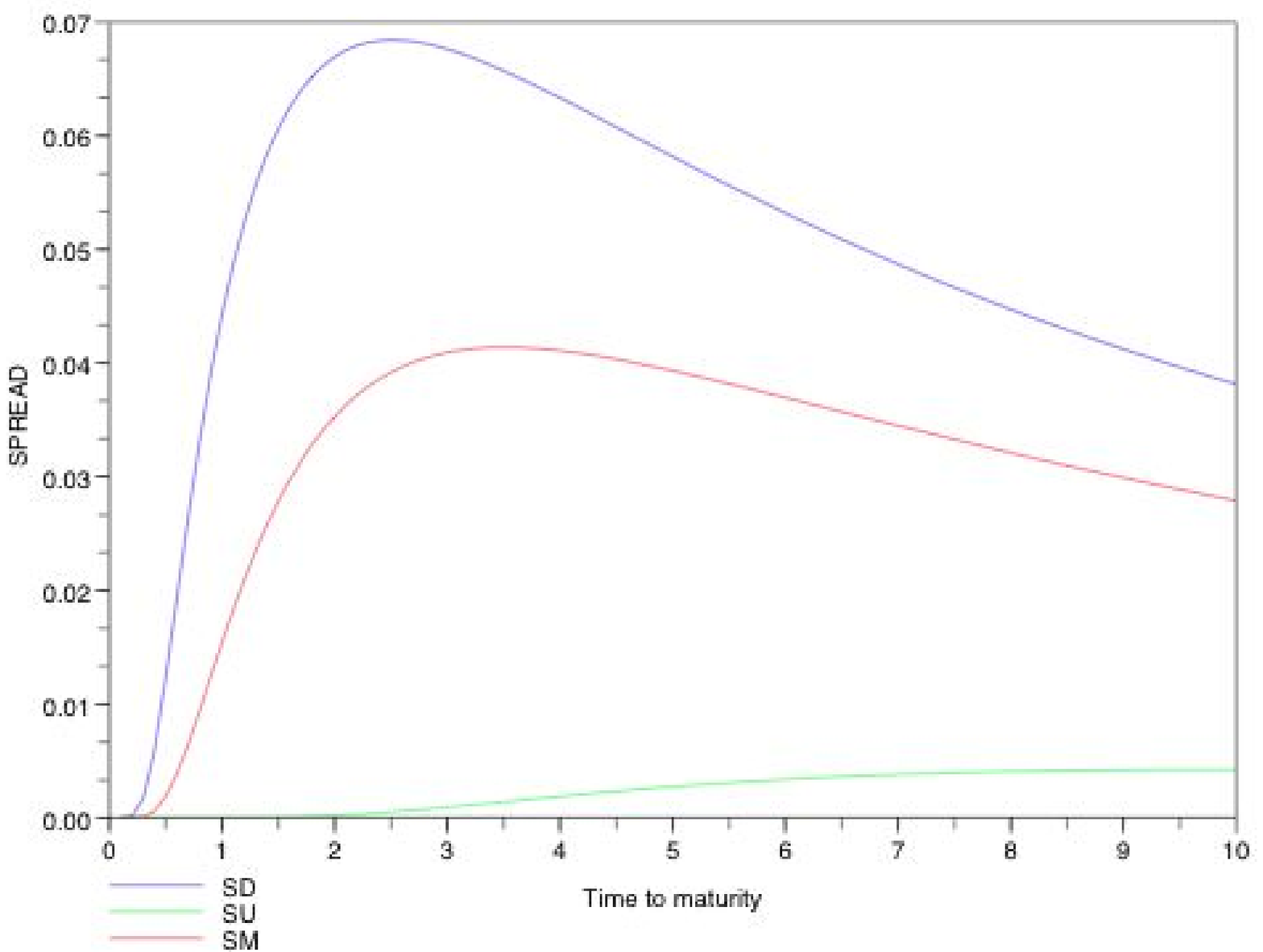}
  \caption{\small{Three trajectories of the value process in the full information case (on the left) and the corresponding spreads (on the right).}}
\label{figureVBS}
  \end{figure}

\vskip 0.3cm
First of all, we notice that the short term spreads under partial information, being the default time totally inaccessible, do not vanish, as it is the case in the full information model. Moreover, since $V_t$ and $S_t$ are positively correlated, it is expected that  the more the trajectory of $S$ behaves ``badly'', the higher the short term spreads are, as shown in Figure  \ref{figureSBS}.

In the full information setting the short term spreads are always equal to zero, but in ``bad'' situations (for example in the case of the blue trajectory on the left-hand side of Figure \ref{figureVBS}) the medium term spreads can be higher than in the partial information model.

The credit spreads under partial information can, then, be considered as good indicators of the transparency of a firm as viewed by bond-market participants.

\vskip 0.3cm
\noindent  $\rhd$ {\it The CEV model.} We suppose that the firm value's and the observed process' dynamics are given by
\begin{equation}
\left \{
\begin{array}{rclr}
d V_t &=& V_t ( \mu dt + \gamma V_t^{\beta} dW_t), & V_0 = v_0,\\
d S_t & = & S_t (r dt + \sigma d W_t + \delta d \bar W_t), &S_0=v_0,
\end{array}
\right.
\end{equation}
where $\mu=r=0.03$, $\gamma=744.7$ (it is chosen so that the initial volatility equals $0.10$), $\beta=-2$ (notice that in this case one of the characteristics of the model is that the \emph{leverage} effect holds: a firm value process increase implies a decrease in the variance of the price process return), $\sigma=0.05, \delta=0.1, v_0=86.3$.
\begin{Rem}{\rm
Since in this situation the firm value process reaches zero with a positive probability (see, e.g., Jeanblanc and Yor, 2009), we have to work with the stopped process ${(V_{\tau \wedge t})}_t$.
}
\end{Rem}
Numerical results are presented in Figure \ref{figureCEV}, where three trajectories of the observed process $S$ and the corresponding spreads are depicted.
We first notice that the spreads in this example are higher than the ones in the previous example. This is due to the fact that in this case the observed process $S$ is more volatile, as it can be seen from Figure \ref{figureCEV}, compared to Figure \ref{figureSBS}.

Secondly, we remark, as in the previous example, that the more the trajectory of S behaves ``badly'', the higher the short
term spreads are, as shown in Figure \ref{figureCEV}.

Moreover, notice that the spread curves corresponding to the two worst $S$ trajectories seem to cross, however a zoom in the graph shows that it is not the case and that the blue curve is always above the red one. This can be explained by noticing that the model we use keeps the memory of all the observed path and that the blue trajectory of $S$ is globally worse than the red one.
\begin{center}
\begin{figure}[htpb]
\includegraphics[width=8.3cm,height=6.0cm]{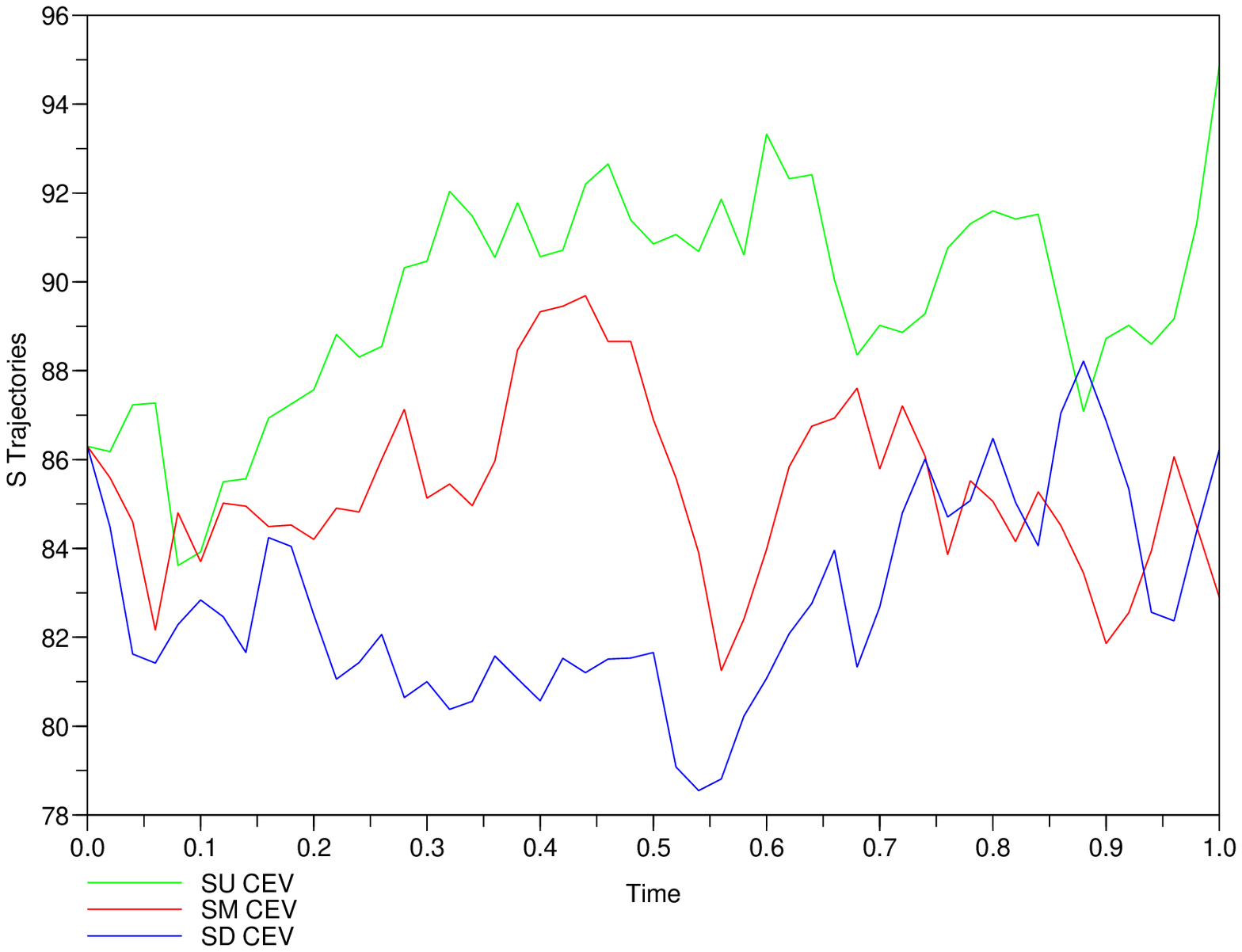}
\hfill \includegraphics[width=8.3cm,height=6.0cm]{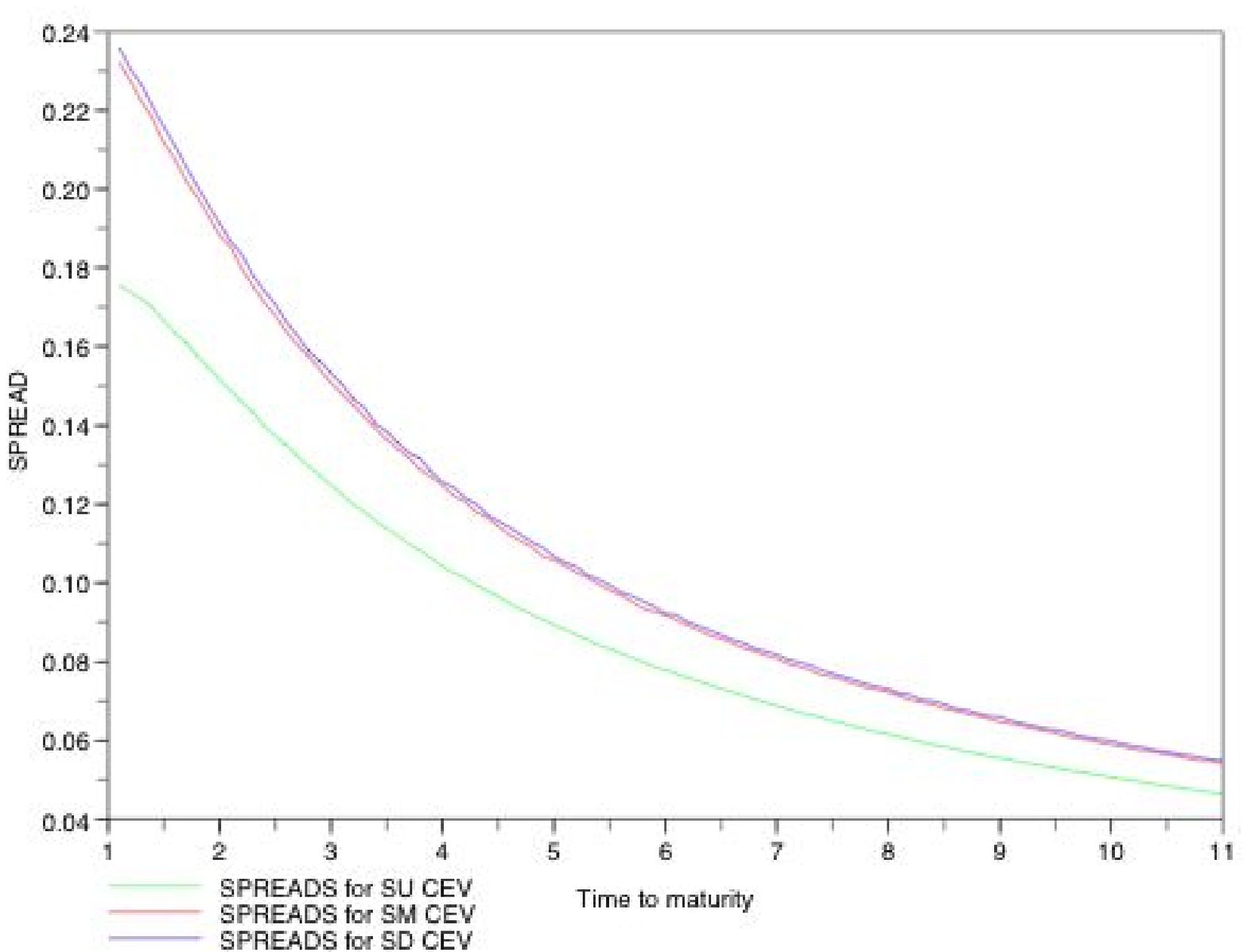}
\caption{\small{Three trajectories of the observed process $S$ in the CEV model (on the left) and the corresponding spreads (on the right).}}
\label{figureCEV}
\end{figure}
\end{center}


We end this section by pointing out some general facts concerning the me\-tho\-do\-lo\-gy we have used.
\begin{Rem} {\rm
a) The most  important fact from the numerical point of view is that, as soon as the process $V$ is quantized over $[0,s]$, the survival probability   $\mathbb{Q}(\inf_{s \le u \leq t} V_u >a \vert \mathcal F_s^{S})$ is estimated for every maturity  $t>s$ without modifying the optimal quantization grid of $V$.\\
\noindent
b)  Notice also that this method works as soon as the process $V$ can be quantized.\\
\noindent
c) Dealing with the computation of credit spread we have to work under the probability $\mathbb Q$. In the case where $S$ is traded in the market the drift under the measure $\mathbb Q$ is equal to $r$ and we do not need to know the function $\psi$ in order to estimate the survival probability. When the observed process is not traded in the market, on the contrary, and if we are not interested in the computation of the credit spreads, the survival probability can be computed under the historical probability $\mathbb P$, so that the function $\psi$ is exploited.
}
\end{Rem}

\vskip 1cm
\section*{References}
\bi
\item[] \hspace{- 1.05 cm} Bielecki, T.R., \& Rutkowski, M. (2004). {\it Credit Risk: Modeling, Valuation and Hedging}. Berlin: Springer (Springer Finance). \vspace{- 0,2 cm}
\item[] \hspace{- 1.05 cm} Borodin, A.N., \& Salminen, P. (2002). {\it Handbook of Brownian motion - Facts and Formulae}. 2nd ed. Basel: Birkh\"{a}user. \vspace{- 0,2 cm}
\item[] \hspace{- 1.05 cm} Coculescu, D., Geman, H., \& Jeanblanc, M. (2008). Valuation of Default Sensitive Claims under Imperfect Information. {\it Finance and Stochastics}, 12(2), pp. 195-218. \vspace{- 0,2 cm}
\item[] \hspace{- 1.05 cm} Duffie, D., \& Lando, D. (2001). Term Structures of Credit Spreads with Incomplete Accounting Information. {\it Econometrica}, 69(3), pp. 633-664. \vspace{- 0,2 cm}
\item[] \hspace{- 1.05 cm} Frey, R., \& Runggaldier, W.J. (2009). Nonlinear filtering in models for interest rate and credit risk. To appear in {\it Handbook on nonlinear filtering}. Oxford University Press. \vspace{- 0,2 cm}
\item[] \hspace{- 1.05 cm} Gersho, A., \& Gray, R. (1992). {\it Vector Quantization and Signal Compression}. Boston: Kluwer Academic Press. \vspace{- 0,2 cm}
\item[] \hspace{- 1.05 cm} Gobet, E. (1998). {\it Sch\'{e}mas d'Euler pour diffusion tu\'{e}e. Application aux options barri\`{e}re}. PhD thesis. Paris: Universit\'{e} Paris VII. \vspace{- 0,2 cm}
\item[] \hspace{- 1.05 cm} Gobet, E. (2000). Weak approximation of killed diffusion using Euler schemes. {\it Stochastic Processes and their Applications}, 87, pp. 167-197. \vspace{- 0,2 cm}
\item[] \hspace{- 1.05 cm} Graf, S., \& Luschgy, H. (2000). {\it Foundations of Quantization for Probability Distributions}. Volume 1730 of Lecture Notes in Mathematics. Springer - Verlag. \vspace{- 0,7 cm}
\item[] \hspace{- 1.05 cm} Jarrow, R.A., \& Protter, P. (2004). Structural versus reduced form models: a new information based perspective. {\it Journal of investment management}. 2(2), pp. 1-10. \vspace{- 0,2 cm}
\item[] \hspace{- 1.05 cm} Jeanblanc, M., Yor, M., \& Chesney, M. (2009). {\it Mathematical Methods for Financial Markets}. Springer (to appear). \vspace{- 0,2 cm}
\item[] \hspace{- 1.05 cm} Kahal\'{e}, N. (2007). Analytic crossing probabilities for certain barriers by Brownian motion. Preprint. \vspace{- 0,2 cm}
\item[] \hspace{- 1.05 cm} Kusuoka, S. (1999). A remark on default risk models. {\it Advances in Mathematical Economics}, 1, pp. 69-82. \vspace{- 0,2 cm}
\item[] \hspace{- 1.05 cm} Linestky, V. (2004). The spectral decomposition of the option value. {\it International Journal of Theoretical and Applied Finance}, 7(3), pp. 337-384. \vspace{- 0,2 cm}
\item[] \hspace{- 1.05 cm} Nakagawa, H. (2001). A Filtering Model on Default Risk. {\it J. Math. Sci. Univ. Tokyo}, 8(1), pp. 107-142. \vspace{- 0,2 cm}
\item[] \hspace{- 1.05 cm} Pag\`{e}s, G. (1997). A space quantization method for numerical integration. {\it Journal of Computational and Applied Mathematics}, 89(1), pp. 1-38. \vspace{- 0,2 cm}
\item[] \hspace{- 1.05 cm} Pag\`{e}s, G. (2008). {\it Introduction to numerical probability for finance}. Lecture notes for the Master 2 ``Probabilit\'{e}s et applications'', University Paris VI. \vspace{- 0,2 cm}
\item[] \hspace{- 1.05 cm} Pag\`{e}s, G., \& Pham, H. (2005). Optimal quantization methods for nonlinear filtering with discrete time observations. {\it Bernoulli}, 11(5), pp. 893-932. \vspace{- 0,2 cm}
\item[] \hspace{- 1.05 cm} Pag\`{e}s, G., \& Printems, J. (2003). Optimal quadratic quantization for numerics: the Gaussian case. {\it Monte Carlo Methods and Appl.}, 9(2), pp. 135-165. \vspace{- 0,2 cm}
\item[] \hspace{- 1.05 cm} Pham, H., Runggaldier, W.J., \& Sellami, A. (2005). Approximation by quantization of the filter process and applications to optimal stopping problems under partial observation. {\it Monte Carlo Methods and Applications}, 11(1), pp. 57-81. \vspace{- 0,2 cm}
\item[] \hspace{- 1.05 cm} Revuz, D., \& Yor, M. (1999). {\it Continuous martingales and Brownian motion}. 3rd ed., volume 293 of Grundlehren der mathematischen Wissenschaften [Fundamental Principles of Mathematical Sciences]. Berlin: Springer-Verlag. \vspace{- 0,2 cm}

\ei
\end{document}